\documentclass[a4paper,USenglish]{lipics-v2016}

\usepackage{stmaryrd}
\usepackage{microtype}
\theoremstyle{plain}
\newtheorem{proposition}[theorem]{Proposition}
\newtheorem{claim}[theorem]{Claim}
\newcommand{\NN}{\mathbb{N}}
\newcommand{\ZZ}{\mathbb{Z}}
\newcommand{\QQ}{\mathbb{Q}}
\newcommand{\RR}{\mathbb{R}}
\newcommand{\ACA}{\mathcal{A}} 
\newcommand{\ACN}{\mathcal{N}} 
\newcommand{\ACQ}{\mathcal{Q}} 
\newcommand{\ACC}{\mathfrak{C}} 
\newcommand{\RT}{\operatorname{RT}} 
\newcommand{\CH}{\operatorname{CH}} 

\begin{document}

\title{A Linear Acceleration Theorem for 2D Cellular Automata on all Complete Neighborhoods}
\titlerunning{Linear Acceleation on 2DCA}

\author[1]{Anaël Grandjean}
\author[2]{Victor Poupet}
\affil[1]{LIRMM, Université Montpellier\\
161 rue Ada, 34392 Montpellier, France\\
\url{anael.grandjean@lirmm.fr}}
\affil[2]{LIRMM, Université Montpellier\\
161 rue Ada, 34392 Montpellier, France\\
\url{victor.poupet@lirmm.fr}}
\authorrunning{A.\,Grandjean and V.\,Poupet}
\Copyright{Anaël Grandjean and Victor Poupet}

\EventEditors{Ioannis Chatzigiannakis, Michael Mitzenmacher, Yuval
Rabani, and Davide Sangiorgi}
\EventNoEds{4}
\EventLongTitle{43rd International Colloquium on Automata, Languages,
and Programming (ICALP 2016)}
\EventShortTitle{ICALP 2016}
\EventAcronym{ICALP}
\EventYear{2016}
\EventDate{July 11--15, 2016}
\EventLocation{Rome, Italy}
\EventLogo{}
\SeriesVolume{55}
\ArticleNo{XXX}

\maketitle

\abstract{Linear acceleration theorems are known for most computational models. Although such results have been proved for two-dimensional cellular automata working on specific neighborhoods, no general construction was known. We present here a technique of linear acceleration for all two-dimensional languages recognized by cellular automata working on complete neighborhoods.}

\section{Introduction}

Cellular automata (CA) were initially introduced by S.~Ulam and J.~von Neumann \cite{neumann66} in the 1960s to study self-reproduction in discrete dynamical systems. They are massively parallel systems consisting of an infinite array of cells. Cells evolve synchronously depending on the states of their neighbors according to a uniform deterministic rule. Although initially considered in two dimensions, the definition can be adapted to any dimensional cellular space and even more general uniform graphs \cite{roka99}.

Soon after their introduction, they were shown to be computationally universal \cite{smithIII71, albert87}. As a computation model, they have been extensively studied as one-dimensional language recognizers \cite{smithIII72, cole69} but are also very well suited to the study of two-dimensional ``picture languages'' \cite{smithIII71a, terrier99, terrier04}.

The neighborhood of a cellular automaton defines the underlying communications graph of the cells. Although most of the existing work on two-dimensional cellular automata focuses on the von Neumann (4 closest neighbors) and Moore (8 closest neighbors) neighborhoods, understanding how the choice of the neighborhood affects the algorithmic capabilities of the model is a key to understanding parallel computation.


Linear acceleration theorems are well known for most of the commonly considered computation models. It was first proved for one-dimensional cellular automata working on the standard neighborhood and two-dimensional cellular automata on von Neumann's neighborhood by W.~T.~Beyer \cite{beyer69}, inspired by similar constructions for sequential input cellular automata \cite{cole69, fischer65, hennie61}. The one-dimensional case was later generalized by J.~Mazoyer and N.~Reimen for arbitrary neighborhoods \cite{mazoyer92}. As for two-dimensional neighborhoods, V.~Terrier extended the construction to the Moore neighborhood \cite{terrier99} and then to the slightly more general class of neighborhoods whose convex hull has at most one vertex in the positive quarter plane \cite{terrier04}.

In this paper, we prove a general linear acceleration result for all complete neighborhoods on two-dimensional cellular automata.

The main theorem is stated in Section~\ref{sec:main} and proved in Sections~\ref{sec:compression} to \ref{sec:total_time}. Sections~\ref{sec:compression} and \ref{sec:simulation} describe the two main elements of the construction (compression of the input and accelerated simulation of the original automaton respectively). Section \ref{sec:transition} presents a technique to perform a sequence of tasks on a cellular automaton without the need for synchronization at the start of each new task, used in the proof of the theorem to combine sections 4 and 6. Although this technique is elementary and has been used in previous publications (a special case was used by W.~T.~Beyer in 1969 \cite{beyer69}), reviews of previous articles seem to indicate that it is not common knowledge. It is therefore presented here in a separate section and stated generally in the hopes that it can be easily reused in future publications.

The construction presented in this article is similar in several ways to previously published constructions, most notably those of V.~Terrier in \cite{terrier04}. Significant improvements include compression of the input in \emph{almost} optimal time (Section~\ref{sec:compression}, specifically Subsection~\ref{ssec:compression_general}) and a more general simulation technique (Section 6).

\section{Definitions}

\subsection{Cellular Automata}

\begin{definition}[Cellular Automaton]
	A \emph{cellular automaton} (CA) is a quadruple $\ACA = (d, \ACQ, \ACN, \delta)$ where
	\begin{itemize}
		\item $d\in\NN$ is the dimension of the automaton~;
		\item $\ACQ$ is a finite set whose elements are called \emph{states}~;
		\item $\ACN$ is a finite subset of $\ZZ^d$ called \emph{neighborhood} of the automaton~;
		\item $\delta: \ACQ^{\ACN} \rightarrow \ACQ$ is the \emph{local transition function} of the automaton.
	\end{itemize}
\end{definition}

\begin{definition}[Configuration]
	A \emph{$d$-dimensional configuration} $\ACC$ over the set of states $\ACQ$ is a mapping from $\ZZ^d$ to $\ACQ$. The elements of $\ZZ^d$ will be referred to as \emph{cells}.
\end{definition}

Given a CA $\ACA = (d, \ACQ, \ACN, \delta)$, a configuration $\ACC\in\ACQ^{\ZZ^d}$ and a cell $c\in \ZZ^d$, we denote by $\ACN_\ACC(c)$ the neighborhood of $c$ in $\ACC$~:
\begin{displaymath}
	\ACN_\ACC(c) : \left\{\begin{array}{rcl}
		\ACN & \rightarrow & \ACQ \\
		n & \mapsto & \ACC(c+n)
	\end{array}\right.
\end{displaymath}

From the local transition function $\delta$ of a CA $\ACA = (d, \ACQ, \ACN, \delta)$, we can define the \emph{global transition function of the automaton} $\Delta: \ACQ^{\ZZ^d} \rightarrow \ACQ^{\ZZ^d}$ obtained by applying the local rule on all cells~:
\begin{displaymath}
	\Delta(\ACC) = \left\{\begin{array}{rcl}
		\ZZ^d & \rightarrow & \ACQ \\
		c & \mapsto & \delta(\ACN_\ACC(c))
	\end{array}\right.
\end{displaymath}
The action of the global transition rule makes $\ACA$ a dynamical system over the set $\ACQ^{\ZZ^d}$. Because of this dynamic, in the following we will identify the CA $\ACA$ with its global rule so that $\ACA(\ACC)$ is the image of a configuration $\ACC$ by the action of the CA $\ACA$, and more generally $\ACA^t(\ACC)$ is the configuration resulting from applying $t$ times the global rule of the automaton from the initial configuration $\ACC$.

\begin{definition}[Quiescent and Permanent States]
	For a given CA $\ACA$, we say that a state $q$ is \emph{quiescent} if a cell in state $q$ remains in this state if all its neighbors are also in $q$. We say that $q$ is \emph{permanent} if a cell in state $q$ remains in that state regardless of the state of its neighbors.
\end{definition}

In this article we will only consider 2-dimensional cellular automata (2DCA). From now on the set of cells will always be $\ZZ^2$.

\subsection{Neighborhoods}

Throughout the article, we use the additive notation for vector sums, the power notation for neighborhood composition and the product notation for scalar product:
\begin{definition}[Vector Sum]
	Given two neighborhoods $\ACN_1$ and $\ACN_2$ and a cell $c\in\ZZ^2$, we define the vector sums $\ACN_1+\ACN_2=\{x+y \mid x\in\ACN_1, y\in \ACN_2\}$ and $c+\ACN_1 = \{c+x \mid x \in \ACN_1\}$.
\end{definition}
\begin{definition}[Neighborhood Powers]
	Given a neighborhood $\ACN$, we define
	\begin{align}
		\ACN^0 &= \{0\} \\
		\forall k\in\NN,\quad \ACN^{k+1} &= \ACN + \ACN^k
	\end{align}
\end{definition}
\begin{definition}[Scalar product]
	Given a neighborhood $\ACN$ and an integer $k\in\ZZ$, we define the scalar product $k\ACN = \{kx \mid x\in\ACN\}$.
\end{definition}

\begin{definition}[Complete Neighborhood]
	A neighborhood $\ACN$ is said to be \emph{complete} if
	\[
		\bigcup_{k\in\NN}\ACN^k = \ZZ^2
	\]
\end{definition}

\begin{definition}[Convex Hull and Convex Neighborhood]
	The \emph{convex hull} of a neighborhood $\ACN$ is the smallest convex polygon $\CH(\ACN)\subset \RR^2$ such that $\ACN\subseteq \CH(\ACN)$. Moreover a neighborhood $\ACN$ is said to be \emph{convex} if it contains all points of integer coordinates in its convex hull: $\ACN = \CH(\ACN) \cap \ZZ^2$.
\end{definition}
\begin{remark}
	If $\ACN$ is a convex neighborhood, $\ACN^p$ is also convex for any $p\in\NN$.
\end{remark}

\subsection{Two-Dimensional Language Recognition}

\begin{definition}[Picture]
	For $n, m\in \NN$ and $\Sigma$ a finite alphabet, an \emph{$(n, m)$-picture} (picture of width $n$ and height $m$) over $\Sigma$ is a mapping
	\begin{displaymath}
		p: \llbracket 0, n-1\rrbracket \times \llbracket 0, m-1\rrbracket \rightarrow \Sigma
	\end{displaymath}

	$\Sigma^{n, m}$ denotes the set of all $(n, m)$-pictures over $\Sigma$ and $\Sigma^{*,*} = \bigcup_{n, m\in \NN} \Sigma^{n, m}$ the set of all pictures over $\Sigma$. A \emph{picture language} over $\Sigma$ is a set of pictures over $\Sigma$.
\end{definition}

\begin{definition}[Picture Configuration]
	Given an $(n, m)$-picture $p$ over $\Sigma$, we define the \emph{picture configuration} associated to $p$ with quiescent state $q_0\notin \Sigma$ as
	\begin{displaymath}
		\ACC_{p, q_0}: \left\{\begin{array}{rcl}
			\ZZ^2 & \rightarrow & \Sigma \cup \{q_0\} \\
			x, y & \mapsto & \left\{\begin{array}{rl}
				p(x, y) & \qquad\textrm{if $(x, y)\in \llbracket 0, n-1\rrbracket \times \llbracket 0, m-1\rrbracket$} \\
				q_0 & \qquad\textrm{otherwise}
			\end{array}\right.
		\end{array}\right.
	\end{displaymath}
\end{definition}

\begin{definition}[Picture Recognizer]
	Given a picture language $L$ over an alphabet $\Sigma$, we say that a 2DCA $\ACA=(2, \ACQ, \ACN, \delta)$ such that $\Sigma \subseteq \ACQ$ recognizes $L$ with quiescent state $q_0\in \ACQ\setminus \Sigma$, accepting state $q_a \in \ACQ$ and rejecting state $q_r\in \ACQ$ in time $\tau: \NN^2 \rightarrow \NN$ if $q_a$ and $q_r$ are permanent states and for any picture $p$ (of size $n\times m$), starting from the picture configuration $\ACC_{p,q_0}$ at time 0, the origin cell of the automaton at time $\tau(n, m)$ is in state $q_a$ if $p\in L$ and state $q_r$ if $p\notin L$.
\end{definition}

\begin{definition}[Real Time]
	Given a complete neighborhood $\ACN$, the real time function $\RT_\ACN: \NN^2 \rightarrow \NN$ associated to $\ACN$ is defined as
\begin{displaymath}
	\RT_\ACN(n, m) = \min\{t \mid \llbracket 0, n-1\rrbracket \times \llbracket 0, m-1\rrbracket \subseteq \ACN^t\}
\end{displaymath}
\end{definition}

\section{The Main Theorem}
\label{sec:main}

Most of the article will be dedicated to the proof of the following theorem

\begin{theorem}[Linear Acceleration]
	\label{theo:main}
	For any complete neighborhood $\ACN$, any real number $\epsilon>0$, any finite alphabet $\Sigma$ and any language $L \subseteq \Sigma^{*,*}$, if $L$ is recognized by a 2DCA working on $\ACN$ in time
	\begin{displaymath}
		(n, m) \mapsto \RT_\ACN(n, m) + f(n, m)
	\end{displaymath}
	for some function $f: \NN^2 \rightarrow \NN$ then $L$ can be recognized in time
	\begin{displaymath}
		(n, m) \mapsto \left\lceil(1+\epsilon)\RT_\ACN(n, m) + \epsilon f(n, m)\right\rceil
	\end{displaymath}
	by a 2DCA with neighborhood $\ACN$.
\end{theorem}

\begin{corollary}
	For any complete neighborhood $\ACN$, any language recognized in time $(n, m) \mapsto k \RT_\ACN(n, m)$ for some $k > 1$ can be recognized in time $(n, m) \mapsto (1+\epsilon)\RT_\ACN(n, m)$ for any real number $\epsilon>0$.
\end{corollary}

To prove Theorem~\ref{theo:main}, we consider a 2DCA $\ACA$ working on a complete neighborhood $\ACN$ and describe the construction of a 2DCA $\ACA'$ working on the same neighborhood that simulates the behavior of $\ACA$ in a way that enables it to recognize the same language as $\ACA$ in a linearly shorter time.

\subsection{Preliminary Remarks}

The following observations will greatly simplify the proof of Theorem~\ref{theo:main}.

\begin{claim}
	\label{cla:constant}
	It is sufficient to prove Theorem~\ref{theo:main} up to an additive constant, meaning that we only need to prove that $L$ can be recognized in time
	\[
		(n, m) \mapsto (1+\epsilon)\RT_\ACN(n, m) + \epsilon f(n, m) + O(1)
	\]
\end{claim}
\begin{proof}
	Consider that we have this weaker result. To get rid of the $O(1)$ term simply choose $\epsilon' < \epsilon$. For any $C>0$, for $(n,m)$ large enough we have
	\[
		(1+\epsilon)\RT_\ACN(n, m) + \epsilon f(n, m) > (1+\epsilon')\RT_\ACN(n, m) + \epsilon' f( n,m) + C
	\]

	The automaton can handle all the finitely many inputs of small size in real time.
\end{proof}

\begin{claim}
	\label{cla:convex}

	It is sufficient to prove Theorem~\ref{theo:main} for all complete convex neighborhoods.
\end{claim}
\begin{proof}
	Consider a complete neighborhood $\ACN$ and let $\ACN'$ be the convex neighborhood having same convex hull as $\ACN$. The real time functions $\RT_{\ACN}$ and $\RT_{\ACN'}$ differ by at most a constant. Moreover a CA on $\ACN$ can simulate the behavior of a CA on $\ACN'$ with a loss of at most a constant number of steps and conversely (see \cite{delacourt07} for more details).

	The property from the theorem therefore translates directly from one neighborhood to the other with at most a constant difference that can be ignored according to Claim~\ref{cla:constant}.
\end{proof}
From now on we will consider that $\ACN$ is a convex neighborhood.

\section{Compression of the Input}
\label{sec:compression}

The first phase of the construction is to compress the input by a factor $k>\frac{1}{\epsilon}$. We want to move the states of the initial configuration towards the origin, packing them in groups of $k\times k$ as illustrated by Figure~\ref{fig:compression_example}.

Although such compressions are relatively simple to perform on the von Neumann and Moore neighborhoods, on a more general neighborhood it is not possible to know in which direction the information should travel to move towards the origin at optimal speed. In general, the optimal travel direction depends on the proportion $\frac{n}{m}$ of the input.

We first show that if the proportion of the input is fixed, compression can be done in optimal time on any complete neighborhood. Then, by performing a finite number of compressions in parallel, each assuming a different proportion, we show that any input is close enough to one of these assumed proportions to be compressed in ``nearly optimal'' time, which will be sufficient for the proof of Theorem~\ref{theo:main}.

\subsection{Compression of an Input of Constrained Proportion}

If the size of the input is known to be of proportion $\frac n m = \alpha$ for some fixed rational $\alpha$, we can perform a compression by a factor $k$ with a neighborhood such as the one illustrated on the top left of Figure~\ref{fig:compression_rules} with $\frac{x}{y} = \alpha$. On such a neighborhood, compressing the input is simply a matter of transferring the states from the top right to the bottom left, packing them in groups of $(1\times k)$, $(k\times 1)$ or $(k\times k)$ when they cannot go any further in one or both directions (see Figures~\ref{fig:compression_rules} and \ref{fig:compression_example}). The compression is completed in time $\lceil\frac{k-1}{k}\RT\rceil$.

\begin{figure}[htbp]
	\centering
		\includegraphics[page=1]{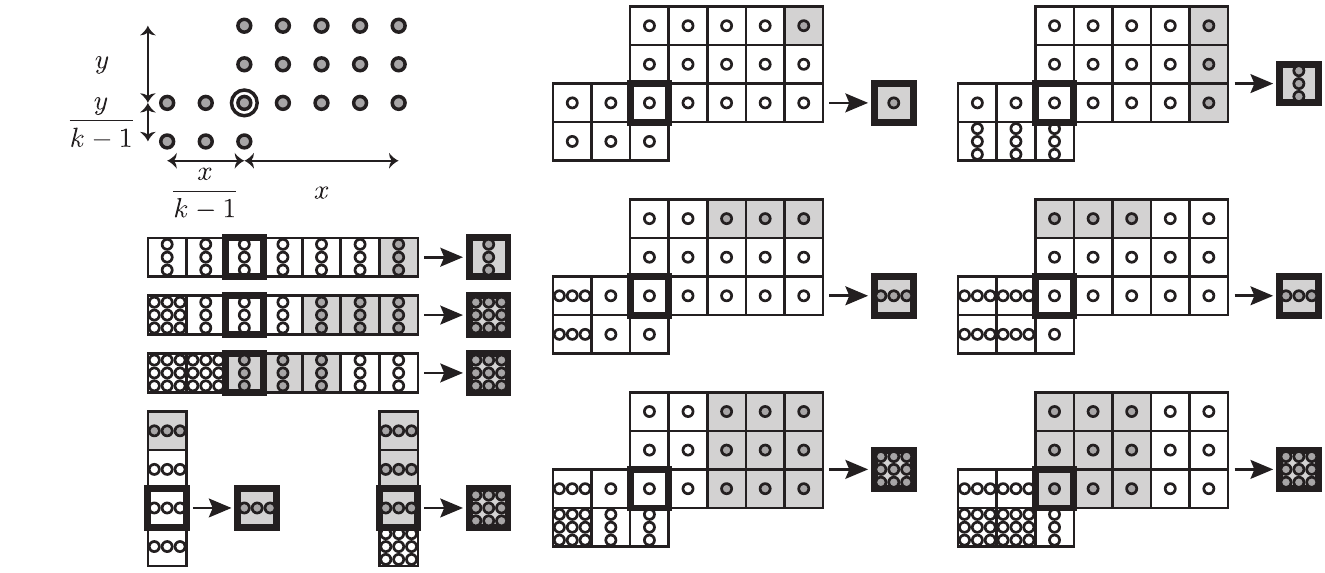}
	\caption{Rules for the compression by a factor 3 of inputs of ratio $\frac{n}{m}=2$ with the neighborhood represented in the top left. The information that the cell takes as its new state is represented in grey. Information travels towards the bottom left. By looking down and left a cell determines if the information from the top right should simply pass through (first case) or if some of its neighbors are already full in which case it should start packing information. Left column shows simplified rules for which the cell has already packed information in one of the directions and therefore only the neighbors in the remaining direction are significant.}
	\label{fig:compression_rules}
\end{figure}

\begin{figure}[htbp]
	\centering
		\includegraphics[page=2]{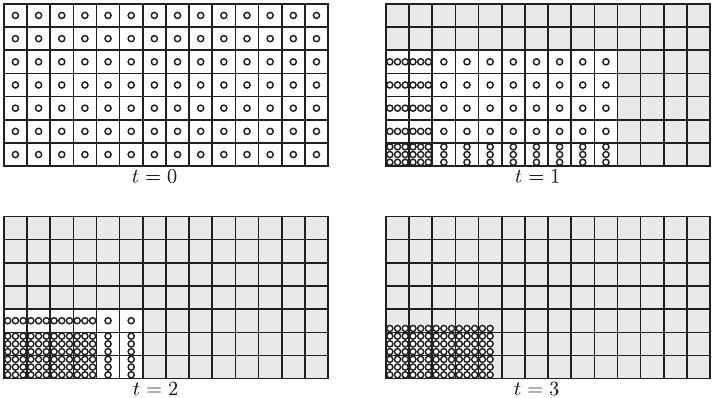}
	\caption{Compression of an input of size $(14\times 7)$ with the neighborhood and rules from Figure~\ref{fig:compression_rules}.}
	\label{fig:compression_example}
\end{figure}

\begin{figure}[htbp]
	\begin{minipage}{0.48\textwidth}
		\centering
			\includegraphics[page=2]{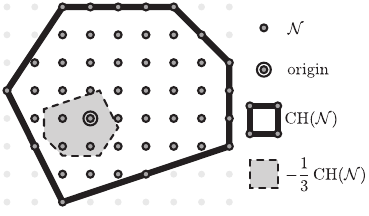}
		\caption{An example neighborhood on which a compression by a factor $k=4$ is possible.}
		\label{fig:homothetic}
	\end{minipage}
	\hfill
	\begin{minipage}{0.48\textwidth}
		\centering
			\includegraphics[page=3]{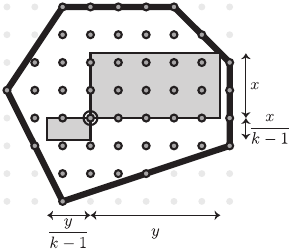}
		\caption{The area of the neighborhood that will be used for the compression of inputs of proportion $\frac{n}{m}=2$ by a factor $k=4$.}
		\label{fig:constrained_neighborhood}
	\end{minipage}
\end{figure}

Note that in order to compress by a factor $k$, the cell must be able to see $\frac{1}{k-1}$ times as far towards the left and bottom as it sees towards the right and top in order to properly determine when it should start packing states. Because the convex hull of a complete neighborhood $\ACN$ contains an open set around the origin (otherwise $(\ACN^p)_{p\in\NN}$ would not expand in all directions), it contains its homothetic image of ratio $-\frac{1}{k_0-1}$ for some $k_0$ (see Figure~\ref{fig:homothetic}). On such a neighborhood compression by any factor $k \geq k_0$ is possible.

To compress inputs of proportion $\frac n m = \alpha$ on some complete neighborhood $\ACN$, we consider the largest rectangle $[0,x]\times[0,y]$ with $x, y\in\QQ$ and $\frac{x}{y}=\alpha$ included in the convex hull of $\ACN$ (see Figure~\ref{fig:constrained_neighborhood}). For all $k\geq k_0$, the rectangle $[-\frac{x}{k-1}, 0]\times [-\frac{y}{k-1}, 0]$ is also in the convex hull of $\ACN$.

These rectangles have rational but not necessarily integer dimensions. If we consider the neighborhood $\ACN^p$ for some large $p$, all is scaled up by a factor $p$ and the corresponding rectangles can be made of integer dimensions. The real time function on inputs of proportion $\alpha$ for $\ACN^p$ is equal to the real time function of the neighborhood containing only the two rectangles (of integer coordinates). The compression algorithm described by Figures~\ref{fig:compression_rules} and \ref{fig:compression_example} therefore finishes in time $\frac{k-1}{k}\RT_{\ACN^p}+O(1)$ on $\ACN^p$. An automaton working on $\ACN$ can simulate one step of an automaton working on $\ACN^p$ in $p$ time steps, and since $\RT_{\ACN^p} = \lceil\frac{1}{p}\RT_{\ACN}\rceil$, the compression can be completed on $\ACN$ in time $\frac{k-1}{k}\RT_{\ACN}+O(1)$.

\subsection{Compression of General Input}
\label{ssec:compression_general}

Let us now consider inputs of arbitrary proportions. As discussed in the previous subsection, for inputs of proportion $\alpha$ the optimal direction in which the information should travel for a compression is defined by the diagonal of the largest rectangle $[0,x]\times[0,y]$ with $\frac x y = \alpha$ included in the convex hull of $\ACN$. The first thing to note is that since $\ACN$ is complete, there exists $\eta>0$ such that $[0,\eta]\times [0,\eta]\subseteq \CH(\ACN)$ and hence all maximal rectangles in $\CH(\ACN)$ have at least one dimension greater than $\eta$. The corners $(x,y)$ of such maximal rectangles all lie on a line. Let us pick a finite set $S$ of rational points on this line from one extremity to the other with distance at most $\epsilon\eta$ between two consecutive points (see Figure~\ref{fig:epsilon}).

For each proportion $\alpha=\frac{x}{y}$ with $(x, y)\in S$, $\ACA'$ performs a compression of the input as described in the previous subsection. All compressions take place at the same time in parallel. Note that even if the proportion of the input is not exactly that for which the compression is optimized, the input is still compressed properly although not as quickly.

Let us prove that one of the compressions that are run by the automaton compresses the input in time at most $(\frac{k-1}{k}+\epsilon)\RT$. A compression along the vector corresponding exactly to the proportion of the input would take a time $\frac{k-1}{k}\RT$. A compression along one of the vectors in $S$ that is closest to the optimal vector (at distance at most $\epsilon\eta$) puts all states from the input within a distance at most $\epsilon\eta\RT$ from their destination in time $\frac{k-1}{k}\RT+O(1)$. By choosing the closest vector properly amongst the two choices, the remaining distance can be travelled in time at most $\epsilon\RT+O(1)$ as illustrated by Figure~\ref{fig:epsilon_travel} (information travels at speed at least $\eta$ in one of the dimensions).

For any possible input, at least one of the compressions completes in time at most $(\frac{k-1}{k}+\epsilon)\RT+O(1)$.

\begin{figure}[htbp]
	\begin{minipage}{.48\textwidth}
		\centering
			\includegraphics[page=1]{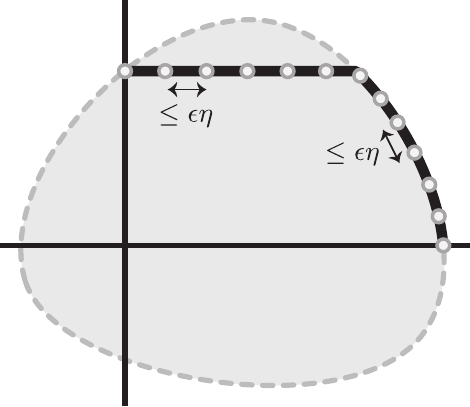}
		\caption{Choosing the set $S$ of proportions that will be used by $\ACA'$ for compressions. The thick line is the set of corners of maximal rectangles for all possible proportions. Along this line, we pick a finite set of points at intervals of at most $\epsilon\eta$.}
		\label{fig:epsilon}
	\end{minipage}
	\hfill
	\begin{minipage}{.48\textwidth}
		\centering
			\includegraphics[page=2]{figures/epsilon.pdf}
		\caption{Compression of an input of arbitrary proportion. The optimal direction for the compression is represented by a dashed line. The closest chosen direction is represented by a solid line inside the neighborhood. It is at a distance at most $\epsilon\eta$ of the optimal vector. The travel path of the farthest cell of the input is represented as a solid black line, made of an initial segment along the almost-optimal direction and an extra segment to compensate for the deviation.}
		\label{fig:epsilon_travel}
	\end{minipage}
\end{figure}

\section{Transition}
\label{sec:transition}

After the input has been compressed, the automaton $\ACA'$ should immediately start simulating the behavior of $\ACA$, $k$ steps at a time. However the cells of $\ACA'$ receive the compressed input at different times. If we wanted all the cells to start the next phase at the same time, we would require some synchronization scheme such as a firing-squad synchronization algorithm but this would take a linear time. Instead, we show that synchronization is not required to start the accelerated simulation as each cell of the automaton proceeds with the next phase as soon as the relevant information is available.

This technique is very general and can be used in numerous situations where a cellular automaton performs a computation by executing a series of separate tasks one after the other without having to spend time synchronizing all cells. In its general form, it can be stated in the following way:

\begin{proposition}[Passive Synchronization]
	\label{prop:synchronization}
	Given a CA $\ACA$ of any dimension working on a complete neighborhood $\ACN$, there exists a CA $\ACA'$ working on the same neighborhood $\ACN$ that can simulate the behavior of $\ACA$ on any input even if the configuration is given asynchronously in such a way that each cell of $\ACA'$ computes states of the simulated automaton at least as fast as if the computation had started synchronously when the last cell receives its input.

	Formally, if we denote by $\ACQ$ the states of $\ACA$, $\ACA'$ has states $\{\bot\}\cup(\ACQ \times \ACQ')$ where $\bot$ is a permanent state (cannot be changed by the transition rule of the automaton) and $\ACQ'$ is a set of extra working states containing a default state $\nu$. The cells of $\ACA'$ are initially in state $\bot$ and considered \emph{inactive}. Before each transition of the automaton, any number of inactive cells of $\ACA'$ might be \emph{activated} by some external action over which $\ACA'$ has no control. Activating a cell $c$ changes its state to $(\ACC(c), \nu)$ where $\ACC$ is the input of the simulated automaton $\ACA$.

	If there exists a time $t_0$ at which all cells have been activated then for any cell $c$ the projection on $\ACQ$ of the state of $c$ in $\ACA'$ at time $(t_0+t)$ is the state of $c$ in the evolution of $\ACA$ from the configuration $\ACC$ at time $t'$ for some $t'\geq t$.

\end{proposition}

\begin{proof}
	The idea is to make all cells of $\ACA'$ compute one step of $\ACA$ whenever they have enough information to do so, while remembering their past states that other cells might need at a later time.

	When a cell $c$ is activated, it receives the initial state $\ACC(c)$ and we say that its \emph{simulated time} is $0$. From that point on, it looks at its neighbors and waits for all of them to be activated. When this happens, it sees all initial states in its neighborhood and can compute the next state in the evolution of $\ACA$, increasing its simulated time to 1. As time passes it keeps watching its neighbors until all of them are also at a simulated time at least equal to its own, which means that it has all the information necessary to compute the next step and increase its simulated time further.

	Let us prove that this process can be carried out with finitely many states. First, notice that since $\ACN$ is complete, there exists $\tau \in \NN$ such that $-\ACN\subseteq \ACN^\tau$. In order to compute its state for a simulated time $(\tau+t)$ a cell $c$ needs to have had access to the state at the simulated time $t$ of all cells in $(c+\ACN^\tau)$ which includes the set $(c-\ACN)$ of cells that have $c$ in their neighborhood. This means that a cell cannot be more than $\tau$ steps ahead in its simulation than the cells that have it in their neighborhood, which implies that the difference of simulated times between two neighbor cells is at most $\tau$. If each cell stores the value of its simulated time modulo $(2\tau+1)$, it is possible for a cell to know the relative difference in simulated time with all cells in its neighborhood. Furthermore, it is sufficient that a cell remembers its last $\tau$ simulated states to be sure that when a cell $c$ at simulated time $t$ looks at its neighbors that are more advanced in the simulation it can see their simulated state at time $t$.

	Finally, we prove by induction that at time $(t_0+t)$ all cells have a simulated time at least $t$. This is obviously true at time $t_0$ since all cells have been activated. By induction, at time $(t_0+t)$ for any cell at simulated time $t$ all its neighbors are at least at simulated time $t$ so it can compute a new step of the simulation, which proves that at time $(t_0+t+1)$ all cells have a simulated time of at least $(t+1)$.

	Note that this process is such that the cells who are behind in their simulation can compute new states without delay, whereas the ones ahead wait for their neighbors to catch up.
\end{proof}

In the following section we describe how $\ACA'$ can simulate $k$ steps of $\ACA$ at a time starting from a compressed input. We assume that all cells complete the compression and start the simulation synchronously when the last cell receives its compressed information. Using Proposition~\ref{prop:synchronization}, we can connect the two constructions (cells are \emph{activated} for the simulation when they receive their compressed input) and ensure that the origin is always at least as advanced in its computation as if the simulation had started synchronously.

\section{Simulation on a Compressed Input}
\label{sec:simulation}

Let us denote by $\rho$ the function that maps a cell of $\ACA'$ to the set of cells of $\ACA$ whose states it receives after the compression of a factor $k$~:
\[
	\forall c\in \ZZ^2,\quad \rho(c) = \{kc + (x, y) \mid x, y \in \llbracket 0, k-1\rrbracket\}
\]
and extend the notation to sets of cells by $\rho(S) = \bigcup\limits_{c\in S}\rho(c)$.

The states of $\ACA$ that are held in the neighborhood of a cell $c$ in $\ACA'$ are the ones corresponding to the cells of $\ACA$ in $\rho(c + \ACN)$. To be able to compute $k$ steps of the original automaton, the cell $c$ in $\ACA'$ needs to be able to see in its neighborhood the states corresponding to the cells $(\rho(c) + \ACN^k)$ of $\ACA$. Although this is the case for simple rectangular neighborhoods, it is not true for some neighborhoods (see Figure~\ref{fig:vonNeumannCE}).

\begin{figure}[htbp]
	\begin{minipage}{.48\textwidth}
		\centering
			\includegraphics[page=1]{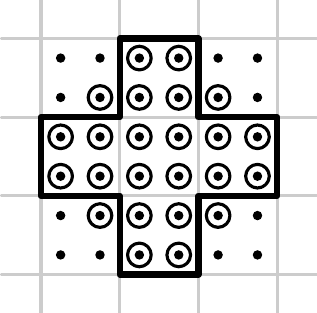}
		\caption{Compression of factor $k=2$ with the von Neumann neighborhood. The thick line represents $\rho(c+\ACN)$, the states that are visible to the central cell in one step. The circles represent $\rho(c)+\ACN^k$, the states that should be known to compute two steps of the original automaton on the states that the central cell holds.}
		\label{fig:vonNeumannCE}
	\end{minipage}
	\hfill
	\begin{minipage}{0.48\textwidth}
		\centering
			\includegraphics[page=2]{figures/vonNeumannCE.pdf}
		\caption{The thick line represents $\rho(c+\ACN)$. The circles represent the cells in $(c'+k\ACN)$ for some $c'$ in $\rho(c)$ (the one in the top left).}
		\label{fig:sim_kN}
	\end{minipage}
\end{figure}

What is true however is that $\rho(c) + k\ACN \subseteq \rho(c + \ACN)$ since for any $v\in\ACN$ if the state of a cell $c'$ in $\ACA$ is held by a cell $c$ in $\ACA'$ after compression of the input, the cell $(c+v)$ in $\ACA'$ holds the state of $(c'+kv)$ in $\ACA$ (see Figure~\ref{fig:sim_kN}).

\begin{lemma}
	\label{lem:alpha_k}
	$\forall k\in\NN, \exists \alpha\in\NN,\quad \ACN^\alpha + k\ACN = \ACN^\alpha + \ACN^k$
\end{lemma}
\begin{proof}
	The inclusion $\ACN^\alpha + k\ACN \subseteq \ACN^\alpha + \ACN^k$ is obvious for any $\alpha$. As for the converse, choose $\alpha$ such that $\alpha + k > |\ACN|(k-1)$. Any $x\in\ACN^\alpha + \ACN^k$ can be written as the sum of $(\alpha + k)$ elements of $\ACN$ and therefore at least one of these elements appears at least $k$ times, which proves that $x\in\ACN^\alpha + k\ACN$.
\end{proof}

By Lemma~\ref{lem:alpha_k}, we can choose $\alpha$ such that $\ACN^\alpha + k\ACN = \ACN^\alpha + \ACN^k$. We now modify the behavior of $\ACA'$ so that during the first $\alpha$ steps of the computation, before starting the compression, all cells gather the initial states contained in their $\ACN^\alpha$ neighborhood. From here onwards, each cell performs all of the computation as described earlier on all the states it holds~: a cell $c$ holds at time $(\alpha+t)$ the states that would have been on all the cells in $(c+\ACN^\alpha)$ at time $t$ on the automaton $\ACA'$ as described until now. At time $(\alpha+t)$ the cells in $(c+\ACN)$ as a whole hold all the states that would have been on the cells in $(c+\ACN^{\alpha+1})$ at time $t$, which is exactly what is needed to compute the states of all cells in $(c+\ACN^\alpha)$ at time $(t+1)$. This extra step adds a constant time $\alpha$ to the computation of the automaton\footnote{The constant time $\alpha$ is actually not lost since the origin holds the states that would be on the cells in $\ACN^\alpha$, which enables it to compute its own state $\alpha$ time steps ahead. However, for the purpose of proving Theorem~\ref{theo:main}, adding a constant time to the computation is irrelevant.}.

After the initial gathering and compression of the input, the cell $c$ in $\ACA'$ holds the initial states in $\ACA$ for the cells in $(\rho(c) + \ACN^\alpha)$. Let us show by induction that this is enough to simulate the behavior of $\ACA'$ with a linear speed-up of factor $k$. Assume that at time $(t_0+t)$, any cell $c$ of $\ACA'$ holds the states at time $t$ for the cells of $\ACA$ in $(\rho(c) + \ACN^\alpha)$.

This means that the cells in the neighborhood $(c+\ACN)$ of $c$ in $\ACA'$ at time $(t_0+t)$ hold the states at time $t$ in $\ACA$ of the cells in $\rho(c+\ACN) + \ACN^\alpha$. By Lemma~\ref{lem:alpha_k}, we have
\[
	\rho(c+\ACN) + \ACN^\alpha \quad \supseteq \quad \rho(c) + k\ACN + \ACN^\alpha \quad \supseteq \quad \rho(c) + \ACN^{\alpha + k}
\]
which shows that cell $c$ in $\ACA'$ at time $(t_0+t)$ sees enough information to compute the states in $\ACA$ for the cells in $(\rho(c) + \ACN^\alpha)$ at time $(t+k)$.

\section{Total time}
\label{sec:total_time}

We have completed the description of the behavior of the automaton $\ACA'$. Let us now evaluate the total time taken to recognize the language $L$ recognized by $\ACA$ in time $\RT_\ACN + f$.

The compression of the input takes a time $(\frac{k-1}{k} + \epsilon)\RT_\ACN + O(1)$. The simulation of $\ACA$ from a fully compressed input takes a time $\frac{1}{k}(\RT_\ACN + f) + O(1)$ for some $k\geq \frac{1}{\epsilon}$, and Proposition~\ref{prop:synchronization} shows that no time is lost by completing the compression asynchronously (the time of the compression is the time at which the last cell is correctly compressed).

The total time for the simulation of $\ACA$ is therefore
\[
(\frac{k-1}{k} + \epsilon)\RT_\ACN + \frac{1}{k}(\RT_\ACN + f) + O(1) \leq (1+\epsilon)\RT_\ACN + \epsilon f + O(1)
\]

By Claim~\ref{cla:constant}, the $O(1)$ term can be eliminated, which concludes the proof of Theorem~\ref{theo:main}.

\section{Conclusion}

The linear acceleration presented in this article is slightly weaker than the previously known results on a limited class of neighborhoods (which contains the von Neumann and Moore neighborhoods). On these neighborhoods, as well as all one-dimensional complete neighborhoods, any language that can be recognized in time $(\RT + f)$ can be recognized in time $(\RT+\epsilon f)$ for any $\epsilon> 0$.

Although the difference is only significant if $f=o(\RT)$, it would be interesting to know whether this stronger statement can be proved for general two-dimensional complete neighborhoods. This would either require an optimal-time compression of the input or a completely different construction skipping the compression altogether.

As we currently understand it, optimal-time compression seems unlikely on general neighborhoods. The problem is that states from the initial configuration should move towards the origin in the optimal direction permitted by the neighborhood. Before receiving any information from the axes, a cell has no way of knowing the precise direction to the origin. If the neighborhood's convex hull has more than one vertex in the positive quarter of the plane, moving along any of the directions permitted by the neighborhood might be  sub-optimal, as oppposed to the case of the Moore neighborhood in which going diagonally at first is never sub-optimal and by the time it is necessary to change direction to go either horizontally or vertically information is received from the axes.

If only one cell needs to send its information towards the origin, the problem can be solved by spreading the information in all directions and spreading symmetric signals from the origin. It is however not possible to implement this for all cells at the same time with finitely many states.

\section*{Acknowledgments}

The authors would like to thank Jacques Mazoyer for his helpful conversations and inspiring ideas at the start of the work that led to this article.

\bibliographystyle{plainurl}
\bibliography{mybib.bib}

\end{document}